\documentclass[12pt]{article}
\usepackage{geometry}
\usepackage{setspace}
\usepackage{amssymb,amsmath, amsthm, bbm, amsfonts, eurosym, graphicx, caption, color, setspace, sectsty, comment, caption, pdflscape, array, hyperref, caption, subcaption, soul, xcolor, tikz-network, natbib, enumitem, booktabs}

\usepackage{algorithm}
\usepackage{algorithmic}

\hypersetup{
  colorlinks,
  citecolor= black,
  linkcolor=black,
  urlcolor=black}

\makeatletter
\newcommand{\succprec}{\mathrel{\mathpalette\succ@prec{\succ\prec}}}
\newcommand{\succ@prec}[2]{\succ@@prec#1#2}
\newcommand{\succ@@prec}[3]{%
  \vcenter{\m@th\offinterlineskip
    \sbox\z@{$#1#3$}%
    \hbox{$#1#2$}\kern-0.4\ht\z@\box\z@
  }%
}
\makeatother

\usepackage[hang, flushmargin]{footmisc}

\DeclareMathOperator{\E}{\mathbb{E}}

\newtheorem{theorem}{Theorem}

\newtheorem{example}{Example}

\newtheorem{proposition}{Proposition}

\setcitestyle{round}

\begin{document}

\title{Making Serial Dictatorships Fair\thanks{I thank Federico Echenique, Peter Hull, Ravi Jagadeesan, Teddy Mekonnen, Bobby Pakzad-Hurson, Roberto Serrano, Neil Takhral, Olivier Tercieux, Rajiv Vohra, Bumin Yenmez, and participants at the Brown Theory Lunch for helpful comments.}
\author{Adam Hamdan\thanks{Department of Economics, Brown University, \texttt{\href{mailto:adam_hamdan@brown.edu}{adam\_hamdan@brown.edu}}}}}
\maketitle

\begin{abstract}
    In priority-based matching, serial dictatorship (SD) is simple, strategyproof, and Pareto efficient, but not free of justified envy (i.e. fair). This paper studies how to optimally order agents in SD based on their priorities to minimize justified envy. I show that this problem can be formulated in terms of rank aggregation and identify a novel connection with the social choice literature: for any distribution of agents' preferences and any capacity structure, the serial order that minimizes the expected number of justified envy cases coincides with the (weighted) Kemeny ranking of agents' priorities \citep{kemeny1959mathematics}. 
    
    \medskip
    \noindent \textbf{Keywords:} Serial dictatorship, priority-based matching, justified envy, rank aggregation, social choice, Kemeny's rule.

    \medskip
    \noindent \textbf{JEL Codes:} C78, D47, D61, D71.
\end{abstract}


\newpage
\section{Introduction}

In priority-based matching problems---such as school choice, public housing allocation, and organ assignment---a set of objects must be allocated to agents in a way that reflects each agent's preferences and respects each object's priority ordering over agents. Two normative principles naturally arise. \textit{Efficiency} requires that no agent can be made better off without making another agent worse off. \textit{Fairness} requires the elimination of justified envy: no agent should prefer another agent's assignment while holding higher priority for that object. Unfortunately, there is no allocation mechanism that satisfies both principles simultaneously \citep{roth1982economics, abdulkadirouglu2003school}, thus highlighting a critical tradeoff between efficiency and fairness.
 
Among mechanisms that are efficient, one of the most prominent in practice is serial dictatorship (SD), along with its non-deterministic counterpart, random serial dictatorship (RSD) \citep{abdulkadirouglu1998random}.\footnote{Variants of SD are used to distribute college dormitories \citep{chen2002improving}, to allocate school seats \citep{abdulkadirouglu2017regression}, to match bureaucrats with their preferred postings \citep{khan2019making}, to assign legislative positions based on priority \citep{sonmez2024constitutional}, to match cadets with their army branches of choice \citep{sonmez2013matching, greenberg2024redesigning}, and to allocate undergraduate courses \citep{kornbluth2024undergraduate}.} In SD, agents are first ordered arbitrarily (or uniformly at random), and in turn, choose their favorite object among the options that remain. In addition to being efficient, SD is strategyproof (participants have an incentive to truthfully reveal their preferences), simple to implement, and has been characterized as uniquely optimal in a host of allocation problems \citep{svensson1999strategy, papai2001strategyproof, ehlers2003coalitional, bade2015serial, pycia2024random}. Crucially, however, since the order in which agents are matched neglects the objects' priorities, the matching produced by SD is generally not free of justified envy. This shortcoming serves as one of the primary motivations in the seminal  work of
\cite{abdulkadirouglu2003school} on school choice: 
\begin{quote}
    This mechanism [RSD] is not only \textit{Pareto efficient}, but also \textit{strategyproof} (i.e., it cannot be manipulated by misrepresenting preferences), and it can accommodate any hierarchy of seniorities. So why not use the same mechanism to allocate school seats to students? The key difficulty with this approach is the following: Based on state and local laws, the priority ordering of a student can be different at different schools.
\end{quote}

Although SD does not eliminate justified envy in settings like school choice, a planner may still wish to use it for its other desirable properties. If so, the natural question is: how can she make SD as fair as possible? In particular, based on the schools' priorities, how should agents be ordered in SD to minimize justified envy? 

The ability to choose a serial order as a function of schools’ priorities rests on the fact that those priorities are dictated by law, and cannot, therefore, be manipulated.\footnote{This would not be true in a college admissions model, where both student and college preferences are subject to misreporting. See \cite{erdil2008s}: ``The important difference between the two models is that in school choice, the priority rankings are determined by local (state or city) laws and education policies, and do not reflect the school preferences, whereas in the college admissions model these rankings correspond to college preferences.''} By the same logic, however, an agent's own position in SD cannot depend on his self-reported preferences without compromising strategyproofness. To circumvent this, the planner can commit not to condition her choice of serial order on agents' preferences. Instead, she takes the distribution of preferences as given, and compares candidate serial orders based on the expected number of justified envy cases they generate. What remains, in essence, is a rank aggregation problem, where the goal is to combine multiple object-specific priority rankings into a single ranking of agents.

The central insight of this paper establishes a novel connection between justified envy minimization in SD and a classical concept from the rank aggregation literature known as the \textit{Kemeny ranking}---the permutation that minimizes the sum of pairwise disagreements from a profile of individual rankings \citep{kemeny1959mathematics}.\footnote{For a characterization of Kemeny's rule in social choice, see \cite{young1978consistent}.} Specifically, my first main result pins down conditions under which this connection is exact: if preferences are identical across agents and uniformly distributed, and objects have unit capacities, the serial order that minimizes the expected number of justified envy cases \textit{is} the Kemeny ranking of agents' priorities (\hyperref[theorem1]{Theorem 1}). 

The intuition for \hyperref[theorem1]{Theorem 1} is as follows. Under identical preferences, envy is guaranteed, and the only remaining avenue is to limit the ``justified'' component of all justified envy cases. Moreover, when preferences are drawn in such a way that every object is equally likely to be envied by every pair of agents, it is optimal to value each object's priority ranking of those agents equally, and that is precisely what the Kemeny ranking achieves. 

I then extend the analysis to settings with non-identical preferences, non-uniform preference distributions, and non-unit capacities, and show that the optimal serial order is always a weighted Kemeny ranking that accounts for how likely each priority violation is to create justified envy.\footnote{For a recent characterization of weighted Kemeny rankings, see \cite{lederer2024bivariate}.} To illustrate this, I relax the assumptions in \hyperref[theorem1]{Theorem 1} one at a time and derive closed-form expressions for the weights.

First, I consider what happens when the identical preference ranking is no longer uniformly distributed. Naturally, if some objects are more likely to appear at the top of the preference ranking, assignments to these objects are likely to occur earlier in the SD sequence. In response, the optimal serial order should incorporate the top-portion of these objects' priority rankings more than others'. Concretely, this amounts to obtaining a weighted Kemeny ranking, where an object's disagreement at the $t$-th position of the serial order is weighted proportionately to how likely that object is to appear at the $t$-th position of the preference ranking (\hyperref[proposition1]{Proposition 1}). 

Second, I characterize the optimal serial order when agents' preferences are no longer identical, but independently drawn. In such cases, an agent is not guaranteed to feel envy towards an earlier agent's assignment. Instead, the probability of envy is increasing in the gap between agents' positions. Once again, the optimal SD mechanism solves a weighted version of Kemeny's problem that assigns more importance to disagreements towards the bottom of the ranking, and accounts for the severity of disagreements at each agent pair (\hyperref[proposition2]{Proposition 2}).

Finally, I allow objects to have non-unit capacities (e.g. schools have multiple seats). The key complication here is that the probability of envy may not be strictly increasing in agents' positions. For example, if every object has at least two copies, the second dictator will never be envious of the first dictator's assignment. Accordingly, the optimal serial order takes the form of a Kemeny ranking with object- and position-specific weights that jointly capture (i) the likelihood of matching with each object, and (ii) the likelihood of reaching full capacity for that object at different steps of the SD matching (\hyperref[proposition3]{Proposition 3}). 

Taken together, these results illustrate that minimizing justified envy in SD consistently reduces to a weighted Kemeny ranking problem.

\bigskip
\noindent \textbf{Related Literature.} This paper contributes to the vast literature on priority-based allocation, starting with the work of \cite{abdulkadirouglu2003school}. Among the earliest attempts to reconcile fairness with efficiency in these problems is work by \cite{ergin2002efficient} characterizing priority domains for which the deferred acceptance (DA) mechanism is both stable and efficient. Another strand of the literature has considered relaxing the goal of no justified envy, particularly in the school choice context (see \cite{kesten2010school, dur2019school, ehlers2020legal, reny2022efficient}). A more recent approach---initiated by \cite{abdulkadiroǧlu2020efficiency}---is to ask which mechanism is justified envy minimal within a class of rules that guarantee efficiency \citep{abdulkadiroǧlu2020efficiency, abdulkadiroglu2020efficient, dougan2022robust, kwon2026justified, afacan2025improving}. The current paper answers this question for the class of strategyproof SD mechanisms. 

Aside from fairness, one dimension where many mechanisms fall short is strategic simplicity. In this regard, SD stands apart from other mechanisms like DA and top trading cycles (TTC) by being \textit{obviously strategyproof}: even agents who cannot engage in contingent reasoning have no incentive to misrepresent their preferences \citep{li2017obviously, ashlagi2018stable, troyan2019obviously, pakzad2023stable}.\footnote{In fact, SD satisfies the stronger notion of \textit{strongly obviously strategyproof} \citep{pycia2023theory}, and is minimally complex for strategyproof allocation without transfers \citep{nagel2023measure}.} Practically, simple mechanisms limit the inequity that arises from uneven strategic sophistication among participants \citep{pathak2008leveling}. For this reason, my results offer intuitive fairness improvements to commonly-used SD mechanisms, such as those with an arbitrary or randomly drawn order. 

This paper also adds to recent work analyzing optimal SD mechanisms. This includes \cite{abdulkadiroglu2020efficient}, who characterize fair sequential dictatorships using a set-inclusion notion of justified envy pairs, and \cite{caragiannis2024optimizing}, who optimize serial orders for utilitarian welfare rather than fairness. Unlike both of these papers, I use cardinal measures of justified envy and adopt a rank aggregation perspective. By doing so, this paper is the first to recognize the applicability of Kemeny's rule to matching problems.

Conceptually, my approach is similar in spirit to a contemporaneous paper by \cite{aryal2025desirable} which asks: given a profile of student preferences and an \textit{existing} matching $\mu$, how could we infer the true ranking of schools' desirability? Mine is a fundamentally different question: how should we aggregate schools' priorities into a single ranking of students when the goal is to \textit{determine} a fair matching through SD? 

\bigskip
\noindent \textbf{Outline.} The rest of the paper is organized as follows. \hyperref[section2]{Section 2} presents a standard model of priority-based matching. \hyperref[section3]{Section 3} offers motivating examples. \hyperref[section4]{Section 4} covers the main result and its extensions. \hyperref[section5]{Section 5} concludes.


\section{Model}\label{section2}

A priority-based matching problem (or ``school choice'' problem) is made up of a tuple $(I, S, P, \succ)$, where $I = \{i_1, \ldots, i_n\}$ is a finite set of agents (students), $S = \{s_1, \ldots, s_m\}$ is a finite set of objects (schools), $P = (P_{i_1}, \ldots, P_{i_n})$ is a profile of strict agent preferences (complete, transitive, and irreflexive relation) over $S \cup \{i\}$ for each $i$, and  $\succ = (\succ_{s_1}, \ldots, \succ_{s_m})$ is a profile of strict priority orderings over agents. We write $a P_i b$ to say that agent $i$ prefers $a$ over $b$, and $i \succ_s j$ to say that $i$ has higher priority than $j$ for $s$. Unless specified otherwise, it is assumed that objects have unit capacities and that there are as many agents as there are objects. For expositional convenience, I assume that all objects are acceptable to agents (i.e. $s P_i \{i\}$ for all $s$ and $i$).

Under these assumptions, a matching is a one-to-one mapping $\mu: I \to S$. A matching $\mu$ Pareto dominates another matching $\mu'$ if $\mu(i) P_i \mu'(i)$ for some $i \in I$ and $\mu(i) R_i \mu'(i)$ for all $i \in I$, where $R_i$ denotes the ``at least as good as'' relation associated with $P_i$. A matching $\mu$ is Pareto efficient if no other matching Pareto dominates it. Given a matching $\mu$, we say that agent $i$ has justified envy towards a match $s = \mu(j)$ if $i$ prefers $s$ to his own match $\mu(i)$ and $s$ assigns higher priority to $i$ than $j$. In such cases, the triplet $(i, (j,s))$ constitutes a case of justified envy (or blocking triplet), and the pair $(i,s)$ constitutes a justified envy pair (or blocking pair) comprised of a blocking agent $i$ and blocking object $s$. 

A mechanism $\varphi$ selects a matching for each problem $(I, S, P, \succ)$, where $\varphi(P, \succ)$ denotes the matching $\mu$ selected by $\varphi$, and $\varphi_i(P, \succ)$ denotes the object assigned to agent $i$ under $\varphi$. A mechanism is Pareto efficient if it only selects Pareto efficient matchings, and it is free of justified envy if it only selects matchings with no cases of justified envy. A mechanism is strategyproof if agents cannot benefit from misreporting their true preferences. Formally, $\varphi$ is strategyproof if there is no problem featuring an agent $i$ and an alternative preference $P_i'$ such that $\varphi_i(P_i', P_{-i}, \succ) P_i \varphi_i(P, \succ)$.

The serial dictatorship (SD) mechanism operates as follows: given a preference profile $P$, and an ordering of agents $\rho$, match the highest-ranked agent, $\rho(1)$, to his favorite object in $S$, then match the second highest-ranked agent, $\rho(2)$, to his favorite object in $S \setminus \mu(\rho(1))$, and so on. Such an ordering $\rho$ is called a serial order. Note that SD is both Pareto efficient and strategyproof for any choice of $\rho$. For ease of notation, let $s(t) = \mu(\rho(t))$ denote the object assigned to the agent in position $t$ of SD, and let $S_t = S \setminus \{s(1), \ldots, s(t-1)\}$ denote the set of unmatched objects at step $t$ of SD. Finally, let $\mu_\rho^{\text{SD}}(P, \succ)$ denote the matching prescribed by SD using $\rho$ for problem $(I, S, P, \succ)$, and let $N(\mu_\rho^{\text{SD}} (P, \succ))$ denote the number of justified envy cases arising in $\mu_\rho^{\text{SD}}(P, \succ)$. We say that $\rho$ is \textit{optimal} if it minimizes $\E_P N(\mu_{\rho}^{SD}(P, \succ))$, where the expectation is taken with respect to the distribution of preferences.


\section{Examples}\label{section3}

In a generic problem $(I, S, P, \succ)$, the amount of justified envy in $\mu_\rho^{\text{SD}}(P,\succ)$ jointly depends on the preferences $P$, the priorities $\succ$, and the choice of serial order $\rho$. 

If the planner were to select $\rho$ as a function of both preferences and priorities, there are several avenues through which she could limit justified envy. For example, if at some step $t$ of SD there is an agent $i$ whose most-preferred object among the remaining options is everyone else's least-preferred object, then selecting $\rho(t) = i$ is weakly optimal to reduce justified envy, simply because no future agent will be envious of $i$'s match. Another way to reduce justified envy is to look for mutually top-ranked agent-object pairs: if at any step $t$ of SD, there is an agent $i$ whose most-preferred object among the remaining options also prioritizes him the most among the remaining agents, then selecting $\rho(t) = i$ generates a match with no subsequent justified envy.\footnote{Notions related to mutually top-ranked pairs frequently appear in the matching literature, including in \cite{reny2021simple, fernandez2022centralized, eeckhout2000uniqueness, clark2006uniqueness}.}

The main issue with these approaches is that the resulting SD mechanism would no longer be strategyproof, since agents could misreport to obtain an earlier position in the serial order.\footnote{The argument is similar to why the Immediate Acceptance (Boston) mechanism is not strategyproof \citep{abdulkadirouglu2003school}.} How then should the serial order be determined if the planner commits not to condition her choice on reported preferences? To see how priorities alone can inform this decision, consider the following situation. 

\begin{example}[Priority Dominance] Let $I = \{1,2,3\}$, $S = \{a,b,c\}$, and let $P$ be unknown.

\begin{center}
\noindent \begin{tabular}{c c c | c c c}
    $P_1$ & $P_2$ & $P_3$ & $\succ_a$ & $\succ_b$ & $\succ_c$ \\
    \hline 
    - & - & - & 2 & 2 & 2 \\
    - & - & - & 3 & 1 & 1 \\
    - & - & - & 1 & 3 & 3 \\
\end{tabular}
\end{center}
\end{example}
\bigskip

Every object in this example prioritizes agent $2$ the most. Consequently, any serial order in which agent $2$ is \textit{not} the first dictator runs the risk of causing agent $2$ to feel justified envy. Conversely, choosing $\rho(1) = 2$ guarantees that no later agent can feel justified envy towards agent $2$'s match: even if they envy agent $2$'s match with some object $s$, they would have a lower priority for $s$. Put differently, for any object $s$ that agent $2$ reports as his most-preferred, $(2,s)$ will constitute a mutually top-ranked pair. Therefore, based solely on objects' priorities, it is clear that $\rho(1) = 2$ weakly dominates any alternative $\rho'(1)$. 

A similar idea appears in \cite{abdulkadiroglu2020efficient}, who propose the notion of \textit{priority dominance} to determine $\rho$. At any step of the serial dictatorship, call an unmatched agent $i$ priority dominated by $j$ if all unmatched objects rank $j$ weakly higher than $i$. Now run SD as follows: at any step $t$ of SD, select any agent who is not priority dominated as $\rho(t)$, and match him to his most-preferred acceptable object. When priorities are strict, \cite{abdulkadiroglu2020efficient} show that this procedure is justified envy-minimal among all SD mechanisms, from the perspective of justified envy pair inclusion. 

There are a few reasons why priority dominance is insufficient in the current setting. First, as stated above, the results in \cite{abdulkadiroglu2020efficient} are derived in terms of justified envy pair inclusion, and do not carry over for cardinal comparisons such as counting the number of justified envy cases, or justified envy pairs. In the same vein, priority dominance is too demanding of a priority-comparison and will usually not be met, thereby failing to narrow down the set of candidate orderings. In the previous example, priority dominance pins down $\rho(1) = 2$, but may not provide guidance for selecting $\rho(2), \rho(3)$. To emphasize this point further, the next example depicts a situation where priority dominance fails to rule out any agent for $\rho(1)$. 

\begin{example}[Unknown Preferences] Let $I = \{1,2,3,4,5\}$, $S = \{a,b,c,d,e\}$, and let $P$ be unknown.

\noindent 
\begin{center}
\begin{tabular}{c c c c c| c c c c c}
    $P_1$ & $P_2$ & $P_3$ & $P_4$ & $P_5$ & $\succ_a$ & $\succ_b$ & $\succ_c$ & $\succ_d$ & $\succ_e$ \\
    \hline 
    - & - & - & - & - & 1 & 5 & 2 & 2 & 4\\
    - & - & - & - & - & 4 & 1 & 3 & 3 & 2\\
    - & - & - & - & - & 3 & 3 & 4 & 1 & 1\\
    - & - & - & - & - & 2 & 4 & 1 & 4 & 5\\
    - & - & - & - & - & 5 & 2 & 5 & 5 & 3\\
\end{tabular}
\end{center}
\end{example}
\bigskip

In order to minimize justified envy, the following possibility should be considered when selecting an entry $t$ in $\rho$: if $\rho(t)$ picks an object that does not rank him the highest among the remaining unmatched agents at step $t$, there is a chance that a later agent with higher priority will experience justified envy toward his match. In this example, the choice $\rho(1) = 2$ results in no justified envy if agent $2$ picks $c$ or $d$ (where he is the highest-priority agent), but is likely to create justified envy if agent $2$ picks $b$ (where he is the lowest-priority agent). 

It follows that a desirable selection procedure should mitigate such risks by ranking agents in descending order of priority, and the challenge lies precisely in determining how to compare agents on the basis of priorities. To appreciate the range of options available, consider applying the following popular rank aggregation methods to the example above (breaking ties in favor of earlier-indices).

\bigskip
\noindent \textbf{Plurality rule.} Place the agent who is ranked first most often at the top of the ranking, and iterate. In our example, this yields $\rho^{\text{Plurality}} = 2,3,1,4,5$.

\medskip
\noindent \textbf{Instant-Runoff rule.} Place the agent who is ranked first least often at the bottom of the ranking, and iterate. In our example, this yields $\rho^{\text{Instant-Runoff}} = 2,1,4,5,3$.

\medskip
\noindent \textbf{Coombs' rule.} Place the agent who is ranked last most often at the bottom of the ranking, and iterate. In our example, this yields $\rho^{\text{Coombs}} = 1,3,4,2,5$.

\medskip
\noindent \textbf{Copeland's rule.} Rank agents in descending order of how many other agents they defeat by a majority in pairwise comparison, $c(i) = \sum_{j \in I} \mathbbm{1}\left[\sum_{s \in S} \mathbbm{1}[i \succ_s j] > |S|/2 ] \right]$. In our example, this yields $\rho^{\text{Copeland}} = 1,2,3,4,5$.

\medskip
\noindent \textbf{Borda's rule.} Rank agents in descending order of their Borda scores, summing the number of agents ranked below them across all objects, $b(i) = \sum_{s \in S} \sum_{j \in I} \mathbbm{1}[i \succ_s j]$. In our example, this yields $\rho^{\text{Borda}} = 1,2,4,3,5$.

\medskip
\noindent \textbf{Kemeny's rule.} Select a ranking $\rho$ that minimizes the sum of pairwise disagreements across all the objects' priority rankings, $\rho \in \arg \min_{\rho \in \mathcal{R}_n} \sum_{s \in S} \sum_{t < t'} \mathbbm{1}[\rho(t') \succ_s \rho(t)]$. In our example, this yields $\rho^{\text{Kemeny}} = 2,1,3,4,5$.

\bigskip
When applied to the example above, each one of these rank aggregation methods produces a different ranking of agents. Notice also that all of these methods (correctly) agree with priority dominance in the case where all objects have the same priority ranking over agents.  

\section{Results}\label{section4}

The objective of this paper is to identify the best way of ordering agents in SD based on their priorities. Since the planner commits not to condition her choice on agents' reported preferences, she takes the distribution of preferences as given and compares the expected number of justified envy cases associated with each candidate order. The optimal serial order jointly depends on the priorities and the distribution of preferences. 

\bigskip
\noindent \textbf{Baseline.} As a focal benchmark, I derive the optimal serial order when agents share an identical preference ranking that is drawn uniformly at random, and when objects have unit capacities. This serves as a useful baseline for two reasons. First, the case of identical ordinal preferences frequently appears in the literature, partially driven by the fact that preferences are strongly correlated in settings like school choice (see \cite{abdulkadirouglu2011resolving, burgess2015parents}). Second, it makes the main intuition transparent. With identical preferences, the planner no longer needs to worry about the likelihood of envy (which is now guaranteed at each step of SD), and can focus on limiting the ``justified'' portion of all justified envy cases.

Under these assumptions, it is shown that the optimal SD follows Kemeny's rule. 

\begin{theorem}\label{theorem1}
    If preferences are identical across agents and uniformly distributed, the serial order that minimizes the expected number of justified envy cases is the Kemeny ranking of agents' priorities.
\end{theorem}

\begin{proof}
     Given a problem $(I, S, P, \succ)$, recall that $N\left(\mu_\rho^{\text{SD}}(P, \succ)\right)$ denotes the number of justified envy cases resulting from SD with the order $\rho$, which can be written as
    \begin{align*}
        N\left(\mu_\rho^{\text{SD}}(P, \succ)\right)
        &= \sum_{t = 1}^n \sum_{t' = t + 1}^{n} \underbrace{\mathbbm{1}[ s(t) P_{\rho(t')} s(t')]}_{\text{envy}} \cdot \underbrace{\mathbbm{1}[\rho(t') \succ_{s(t)} \rho(t)]}_{\text{justified}} 
    \end{align*}
    \allowdisplaybreaks
    Taking the expectation of $N\left(\mu_\rho^{\text{SD}}(P, \succ)\right)$ over all identical preference profiles yields
    \begin{align*}
        &\E_P \left[ N\left(\mu_\rho^{\text{SD}}(P, \succ)\right) \right] \\
        &= \sum_{t = 1}^n \sum_{t' = t+1}^{n} \E_P \left[ \mathbbm{1}[ s(t) P_{\rho(t')} s(t') ] \cdot \mathbbm{1}[\rho(t') \succ_{s(t)} \rho(t) ] \right] \\
        &= \sum_{t = 1}^n \sum_{t' = t+1}^{n} \mathbb{P} \left[ [s(t) P_{\rho(t')} s(t') ] \wedge [\rho(t') \succ_{s(t)} \rho(t) ] \right] \\
        &= \sum_{t = 1}^n \sum_{t' = t+1}^{n} \sum_{s \in S} \mathbb{P} \left[ [s(t) P_{\rho(t')} s(t') ] \wedge [\rho(t') \succ_{s(t)} \rho(t) ] | s(t) = s \right] \cdot \mathbb{P}[s(t) = s]  \\
        &= \sum_{t = 1}^n \sum_{t' = t+1}^{n} \sum_{s \in S} \mathbb{P} \left[ s(t) P_{\rho(t')} s(t') | s(t) = s \right] \cdot \mathbb{P} \left[\rho(t') \succ_{s(t)} \rho(t) | s(t) = s \right] \cdot \frac{1}{n}
    \end{align*}
    
    To see that $\mathbb{P}[s(t) = s] = \frac{1}{n}$, note that an agent $\rho(t)$ will match with an object $s \in S$ if and only if $s$ is the $t$-th top-ranked object in the collective preference ranking, which occurs with probability $\frac{(n - 1)!}{n!} = \frac{1}{n}$. It is also clear that if preferences are identical, then $\mathbb{P} \left[ s(t) P_{\rho(t')} s(t') | s(t) = s \right] = 1$ and the optimal serial order is characterized by
    \begin{align*}
        \operatorname*{argmin}_{\rho \in \mathcal{R}_n} \E_P \left[ N\left(\mu_\rho^{\text{SD}}(P, \succ)\right) \right]
        &= \operatorname*{argmin}_{\rho \in \mathcal{R}_n} \frac{1}{n} \sum_{s \in S} \sum_{t < t'} \mathbbm{1}[\rho(t') \succ_s \rho(t)] \\
        &= \operatorname*{argmin}_{\rho \in \mathcal{R}_n} \sum_{s \in S} \mathcal{K}(\rho, \succ_s)
    \end{align*}

    Where $\mathcal{K}(\rho, \succ_s)$ is the Kendall tau distance (number of pairwise disagreements) between $\rho$ and $\succ_s$, and $\operatorname*{argmin}_{\rho \in \mathcal{R}_n} \sum_{s \in S} \mathcal{K}(\rho, \succ_s)$ is the set of Kemeny rankings.
\end{proof}

For context, the Kemeny method was first introduced by \cite{kemeny1959mathematics} as a disciplined way of obtaining a consensus ranking from a profile of individual rankings. It was later characterized by \cite{young1978consistent} as the unique rank aggregation method satisfying the Condorcet principle (suitably generalized to rankings) together with classical notions of consistency and neutrality. Since then, Kemeny's rule has gained a considerable amount of interest outside of social choice---enjoying applications in revealed preference theory \citep{chambers2021recovering}, computer science \citep{dwork2001rank}, and computational biology \citep{lin2010rank}.

In the current setting of priority-based matching, a new application for Kemeny's rule is uncovered. Under the conditions in \hyperref[theorem1]{Theorem 1}, every triplet $(\rho(t'),(\rho(t),s))$ is equally likely to form a justified envy case. As a result, it is optimal to incorporate each object's pairwise ranking (of the form $\rho(t') \succ_s \rho(t)$) equally---or equivalently, to penalize each violation of such a pairwise ranking equally---following Kemeny's rule. Under more general conditions, each violation would have to be weighted by its likelihood of creating justified envy, denoted by $\mathbb{P} \left[ s(t)P_{\rho(t')} s(t') | s(t) = s \right] \cdot \mathbb{P} \left[ s(t) = s\right]$.

In the rest of the paper, I relax the assumptions in \hyperref[theorem1]{Theorem 1} one at a time to illustrate that the optimal serial order generally consists of a weighted Kemeny ranking. All remaining proofs are left to the Appendix.

\bigskip
\noindent \textbf{Non-uniformly distributed preferences.} In many settings, the planner may have some knowledge about which objects are more desirable than others. Maintaining the assumption that agents share an identical preference ranking, this section asks how the optimal serial order changes when that ranking is no longer drawn uniformly at random. Intuitively, when some objects are systematically more likely to appear near the top of the preference ranking, assignments to these objects are likely to occur earlier in the SD sequence, thus generating more subsequent envy. The planner takes this into account when consolidating objects' priority rankings into a single order. 

A convenient way to represent an arbitrary distribution over identical preference rankings is as follows. Let $p(s,t)$ denote the probability that object $s$ occupies position $t$ in the collective preference ranking, such that $\sum_{s \in S} p(s,t) = 1$ for all $t$, and $\sum_{t =1}^n p(s,t) = 1$ for all $s$. To modify \hyperref[theorem1]{Theorem 1}, one simply needs to substitute $p(s,t)$ for the match probability $\mathbb{P}[s(t) = s]$ (which was previously equal to $\frac{1}{n}$). The planner now minimizes a weighted sum of pairwise disagreements, where each disagreement of the form $\rho(t') \succ_s \rho(t)$ (with $t < t'$) is weighted by the probability that $s$ appears in the $t$-th position of the preference ranking. This fact is captured by the following result, whose proof is identical to that of \hyperref[theorem1]{Theorem 1}. 

\begin{proposition}\label{proposition1}
    If preferences are identical across agents and drawn from an arbitrary distribution characterized by probabilities $p(s, t)$, the serial order that minimizes the expected number of justified envy cases is a weighted Kemeny ranking, where each pairwise priority disagreement $\rho(t') \succ_s \rho(t)$ with $t < t'$ receives weight $p(s,t)$.
\end{proposition}

A special case of identical but non-uniformly distributed preferences is if agents have a common preference ranking that is known to the planner. In other words, the planner knows that object $s_1$ is every agent's first choice, $s_2$ is every agent's second choice, and so on. In this case, the optimal serial order in \hyperref[proposition1]{Proposition 1} can be obtained using the following intuitive procedure: at any step $t$, let the $t$-th dictator $\rho(t)$ be the highest-priority agent among the remaining agents according to the most-preferred object among the remaining objects. It is easy to see that in this special case, the optimal serial order completely eliminates justified envy. 

\bigskip
\noindent \textbf{Independent preferences.} How does the optimal SD mechanism adapt when agents' preferences are independently distributed? To repeat the proof of \hyperref[theorem1]{Theorem 1}, two things must be considered anew. First, what is the probability that an agent in position $t$ is matched to a specific object $s$? Second, what is the probability that an agent in position $t'$ is envious of an earlier agent $t$'s assignment to some object $s$?

Addressing the first point, it can be shown that under independent and uniformly drawn preferences, every agent remains equally likely to match with every object, irrespective of his position in $\rho$. To see this, note that the first dictator is equally likely to match with every object when his preference ranking is drawn uniformly at random. Next, the second dictator only matches with an object $s$ if it happens to be his first choice in $S_2$ (which occurs with probability $\frac{1}{n - 1}$) and if it was not assigned to the first dictator (which occurs with probability $\frac{n - 1}{n}$), resulting in a probability of $\frac{1}{n}$. Repeating this argument leads to a match probability of $\frac{1}{n}$ for any $\rho(t)$ and any $s$.

Regarding the second point, independence of preferences will non-trivially affect the conditional probability of envy between agents. Recall that with identical preferences, an agent in position $t'$ is guaranteed to feel envy towards any earlier assignment $s(t)$. Now, envy can arise in multiple ways. For example, an agent occupying the third position in $\rho$ may feel envy towards $\rho(1)$ if $\rho(1)$ happens to pick his first most-preferred object in $S$, but he may also feel envy towards $\rho(1)$ if $\rho(1)$ picks his second most-preferred object and $\rho(2)$ subsequently picks his first most-preferred object. To account for this, the optimal serial order follows a weighted Kemeny ranking. 

\begin{proposition}\label{proposition2}
    If preferences are independently and uniformly distributed, the serial order that minimizes the expected number of justified envy cases is a weighted Kemeny ranking, where each pairwise priority disagreement $\rho(t') \succ_s \rho(t)$ with $t < t'$ receives weight $p(t', t)$ given by
    \begin{align*}
        p(t',t) = \frac{1}{|S_t|} \left( 1 + \sum_{k = 1}^{t' - t - 1} \frac{\binom{t' - t - 1}{k} \ k!}{\prod_{i = 1}^k |S_{t + i}|} \right)
    \end{align*}
\end{proposition}

Now, disagreements of the form $\rho(t') \succ_s \rho(t)$ are weighted more heavily as $t' - t$ gets larger, since $\rho(t')$ becomes more likely to feel envy towards $\rho(t)$'s potential match with object $s$. Holding $t' - t$ fixed, a disagreement is also weighted more heavily when it appears towards the bottom of the ranking. For example, the last dictator $\rho(n)$ has a $\frac{1}{2}$ chance of envying $\rho(n-1)$, while $\rho(2)$ only has a $\frac{1}{n}$ chance of envying $\rho(1)$.

\bigskip
\noindent \textbf{Non-unit capacities.} For some applications (such as school choice) it is natural to assume that objects have non-unit capacities and can potentially accommodate more than one agent. How does the analysis change in many-to-one matching problems? 

First, let us generalize the model in \hyperref[section2]{Section 2} to allow for non-unit capacities. A priority-based matching problem now features a finite set of objects $S$ with $|S| = m$, allowing for $m \neq n = |I|$. Each object $s$ is endowed with a capacity, or quota, $q_s \ge 1$. In the school choice setting, $q_s$ refers to the number of seats available at school $s$. To lighten notation, let $s^k$ denote the $k$-th most preferred object in the collective (identical) preference ranking, let $S^k$ denote the set of objects occupying positions $1,\ldots,k$ in the preferences, where $|S^k| = k$, and let $\binom{S}{k}$ denote the set of all subsets of $S$ of size $k$. 

Now, when evaluating the expected number of justified envy cases arising in SD, the following case-by-case reasoning must be applied: under identical preferences, if agent $\rho(t)$ is matched to an object with unit capacity, then later agents are guaranteed to feel envy towards his match. By contrast, if $\rho(t)$ is matched to an object with multiple copies, $\rho(t+1)$ can still match with that object and will therefore not feel envy towards $\rho(t)$'s match. Extending this argument for any two positions $t < t'$, the optimal serial order now weighs each disagreement $\rho(t') \succ_{s} \rho(t)$ according to (i) how likely it is that $\rho(t)$ matches with object $s$, and (ii) how likely it is that object $s$ reaches full capacity between step $t$ and $t'$. This logic is reflected in the following generalization of \hyperref[theorem1]{Theorem 1} to many-to-one matching problems.

\begin{proposition}\label{proposition3}
    If objects have non-unit capacities and preferences are identical across agents and uniformly distributed, the serial order that minimizes the expected number of justified envy cases is a weighted Kemeny ranking, where each pairwise priority disagreement $\rho(t') \succ_s \rho(t)$ with $t < t'$ receives weight $p(s,t', t, q)$ given by
    \begin{align*}
        p(s, t', t, q) = \sum_{k = 1}^m \sum_{S' \in \binom{S \setminus \{s\}}{k - 1}} \frac{1}{m \binom{m - 1}{k - 1}} \cdot \mathbbm{1}\left[ \sum_{s' \in S'} q_{s'} < t \leq \sum_{s' \in S'} q_{s'} + q_s\right] \cdot \mathbbm{1}\left[ \sum_{s' \in S'} q_{s'} + q_s < t' \right]
    \end{align*}
\end{proposition}

It is easy to verify that \hyperref[proposition3]{Proposition 3} nests \hyperref[theorem1]{Theorem 1} as a special case. If we impose unit capacities $q_s = 1$ for all $s$, and $m = n$, then the indicators collapse to $k = t$ and the above expression simplifies to $p(s, t', t, q) = \frac{1}{n}$ for all $t < t'$. Hence, all priority disagreements receive the same weight, and the optimal serial order reduces to the (unweighted) Kemeny ranking characterized in \hyperref[theorem1]{Theorem 1}.


\section{Discussion}\label{section5}

SD is one of the most canonical allocation mechanisms in economics. For decades, it has been valued by theorists and practitioners for its simplicity, strategyproofness, and efficiency. Yet its failure to eliminate justified envy in settings with multiple priority orderings is regarded as one of its biggest limitations. This paper shows that by optimally selecting the serial order as a function of objects' priorities, SD can be made as fair as possible while preserving its other properties. 

Conceptually, the paper's main contribution is a novel connection between the matching literature and the literature on social choice and rank aggregation. Specifically, I show that if preferences are identical across agents and uniformly distributed, and objects have unit capacities, the optimal serial order coincides exactly with the Kemeny ranking of agents' priorities. More generally, if any of these assumptions are relaxed, the optimal serial order corresponds to a weighted Kemeny ranking.

\bigskip
\noindent \textbf{Relationship with RSD.} In allocation problems where objects have no priority orderings (so-called ``house allocation'' problems), RSD is fair in an ex ante sense by satisfying \textit{equal treatment of equals}: whenever two agents have identical preferences, they receive the same allocation in expectation. It is precisely this symmetric treatment of agents, however, that makes RSD suboptimal in priority-based matching, where fairness is evaluated ex post by the extent of justified envy. 

From this perspective, the optimal SD mechanisms described in this paper may be viewed as generalizations of RSD to environments with priorities. Indeed, when objects have no priorities (equivalently, when each object's priority is the indifference relation over all agents), the set of (weighted) Kemeny rankings coincides with the entire set of $n!$ permutations. By selecting one of these Kemeny rankings uniformly at random, the associated SD mechanism (Kemeny SD) becomes equivalent to RSD. In this sense, Kemeny SD inherits RSD's equal treatment of equals property in the absence of priorities, and minimizes expected justified envy when priorities are non-trivial.

\bigskip
\noindent \textbf{Computational considerations.} In practice, implementing SD with an optimal serial order requires the planner to find a (weighted) Kemeny ranking of agents' priorities---a task which is known to be NP-hard \citep{bartholdi1989voting}. For small matching markets, this remains feasible. However, for applications involving large markets, one needs to look for tractable mechanisms that achieve approximately optimal levels of expected justified envy. 

The computational social choice literature has identified several algorithms to approximate Kemeny rankings, two of which stand out for their simplicity and familiarity: Borda's rule and Copeland's rule, as described at the end of \hyperref[section3]{Section 3} \citep{coppersmith2006ordering, fagin2016algorithmic}. Practitioners can therefore consider ordering agents in SD according to Borda's rule (Borda SD) or Copeland's rule (Copeland SD). Under the conditions of \hyperref[theorem1]{Theorem 1}, the performance bounds of these alternative SD mechanisms can be precisely stated relative to Kemeny SD: If preferences are identical across agents and uniformly distributed, Borda SD has at most five times as many expected justified envy cases as Kemeny SD \citep[][Corollary 3.3]{coppersmith2006ordering}, and Copeland SD has at most four times as many expected justified envy cases as Kemeny SD \citep[][Theorem 4.10]{fagin2016algorithmic}.

\bigskip
\noindent \textbf{Open questions.} To insist on strategyproofness, the planner does not condition her choice of serial order on agents' reported preferences. Without knowledge of the preference profile, justified envy is evaluated in expected terms, and optimality of the mechanism depends on the underlying distribution of preferences. 

One possible alternative is to adopt a distributionally robust perspective. For example, if information about preferences is limited or imprecise, the planner can consider a worst-case approach which assumes that once a serial order is chosen, nature selects preferences to maximize the number of justified envy cases. Another alternative is to allow some features of the serial order to depend on agents' preferences in a way that does not violate strategyproofness. One such avenue is to consider \textit{sequential dictatorships} \citep{papai2001strategyproof, ehlers2003coalitional}, where dictators can be chosen as a function of earlier agents' matches. These extensions are left for future work.

\newpage
\section{Appendix}

\noindent \textbf{Proof of Proposition 2}
    
\begin{proof}
    To adapt the proof of \hyperref[theorem1]{Theorem 1}, we derive new expressions for $\mathbb{P}[s(t) = s]$ and $\mathbb{P}[s(t) P_{\rho(t')} s(t') | s(t) = s]$, under independently and uniformly distributed preferences.
    
    First, we show that $\mathbb{P}[s(t) = s] = \frac{1}{n}$. Fix any $\rho$, and pick some object $s \in S$. There is a $\frac{1}{n}$ chance that the first dictator, $\rho(1)$, picks $s$ as his top choice among $S$. The second dictator, $\rho(2)$, is matched to $s$ if $\rho(1)$ does not pick $s$ in $S_1$ (which occurs with probability $\frac{n -1}{n}$) and if $s$ is $\rho(2)$'s top choice in $S_2$ (which occurs with probability $\frac{1}{n - 1}$). By independence of preferences, the joint event occurs with probability $\frac{n-1}{n} \frac{1}{n-1} = \frac{1}{n}$. For $\rho(3)$, we get $\frac{n - 1}{n} \frac{n - 2}{n - 1} \frac{1}{n - 2} = \frac{1}{n}$. The expression for any $t$ and any $s$ is given by
    \begin{align*}
        \mathbb{P}[s(t) = s] = \left( \prod_{i = 1}^{t - 1} \frac{n - i}{n - i + 1} \right) \frac{1}{n - (t - 1)} = \frac{1}{n}
    \end{align*}  
    Next, we derive $p(t', t) = \mathbb{P}[s(t) P_{\rho(t')} s(t') | s(t) = s]$ directly (conditioning on $s(t) = s$ can be dropped since $\mathbb{P}[s(t) = s] = \frac{1}{n}$). Envy can be decomposed as follows. Fixing a rank $k = 0,\ldots,t' - t - 1$, agent $\rho(t')$ feels envy towards $\rho(t)$ if (i) $s(t)$ is $\rho(t')$'s $(k+1)$-th most-preferred object in $S_t$ (which occurs with probability $\frac{1}{|S_t|}$) and if (ii) the top $1,\ldots,k$-th most-preferred objects in $S_t$ (if any) are assigned to intermediate agents between $t$ and $t'$. There are $\binom{t'-t-1}{k}$ ways to choose which $k$ of these $t'-t-1$ agents receive these objects, and $k!$ ways to assign the $k$ objects to the selected agents in some order. For any fixed assignment and order, the probability that agent $\rho(t+i)$ selects the designated object equals $\frac{1}{|S_{t+i}|}$ (conditional on the object being available in $\rho(t+i)$'s choice set, he chooses randomly). The probability of any such ordered assignment equals $\frac{1}{\prod_{i=1}^k |S_{t+i}|}$. For a fixed $k \ge 1$, event (ii) occurs with probability 
    \begin{align*}
        \frac{\binom{t'-t-1}{k}\,k!}{\prod_{i=1}^k |S_{t+i}|}
    \end{align*}
    Summing the joint probability of events (i) and (ii) over all $k = 0,\ldots,t'-t-1$ yields

    \begin{align*}
        p(t',t) = \frac{1}{|S_t|} \left( 1 + \sum_{k = 1}^{t' - t - 1} \frac{\binom{t' - t - 1}{k} \ k!}{\prod_{i = 1}^k |S_{t + i}|} \right)
    \end{align*}
    Repeating the proof of \hyperref[theorem1]{Theorem 1} with these new expressions yields
    \begin{align*}
    \operatorname*{argmin}_{\rho \in \mathcal{R}_n} \E_P \left[N\left(\mu_\rho^{\text{SD}}(P, \succ)\right)\right] = \operatorname*{argmin}_{\rho \in \mathcal{R}_n} \sum_{s \in S} \sum_{t < t'} p(t',t) \cdot \mathbbm{1}[\rho(t') \succ_s \rho(t)]
    \end{align*}
\end{proof}

\noindent \textbf{Proof of Proposition 3}

\begin{proof}
The goal is once again to derive a new expression for $\mathbb{P}[s(t) P_{\rho(t')} s(t') | s(t) = s]$. Since preferences are identical, agent $\rho(t')$ will be envious of an earlier match $s(t) = s$ whenever $s(t') \neq s(t)$, so $\mathbb{P}[s(t) P_{\rho(t')} s(t') | s(t) = s] = \mathbb{P}[s(t') \neq s(t) | s(t) = s]$, and this event is entirely determined by the remaining capacity of object $s$ at step $t'$ of SD (which itself depends on the capacity of object $s$ when it was assigned to agent $\rho(t)$). 
\allowdisplaybreaks
\begin{align*}
    &\mathbb{P}[s(t') \neq s(t) | s(t) = s] \\
    &= \sum_{k = 1}^m \sum_{S' \in \binom{S \setminus \{s\}}{k - 1}} \mathbb{P}[s^k = s, S^{k - 1} = S'| s(t) = s] \cdot \mathbb{P}[s(t') \neq s(t) | s^k = s, S^{k -1} = S', s(t) = s] \\
    &= \sum_{k = 1}^m \sum_{S' \in \binom{S \setminus \{s\}}{k - 1}} \frac{\mathbb{P}[s^k = s, S^{k - 1} = S', s(t) = s]}{\mathbb{P}[s(t) = s]} \cdot \mathbb{P}[s(t') \neq s(t) | s^k = s, S^{k -1} = S', s(t) = s] \\
    &= \sum_{k = 1}^m \sum_{S' \in \binom{S \setminus \{s\}}{k - 1}}  \frac{\frac{1}{m} \frac{1}{\binom{m - 1}{k - 1}} \mathbbm{1}\left[\sum_{s' \in S'} q_{s'} < t \leq \sum_{s' \in S'} q_{s'} + q_s \right]}{\mathbb{P}[s(t) = s]} \cdot \mathbbm{1}\left[\sum_{s' \in S'} q_{s'} + q_s < t'\right]
\end{align*}
Although $\frac{1}{\mathbb{P}[s(t) = s]}$ eventually cancels out, it is instructive to derive its exact expression. Using the law of total probability, we can evaluate the event $s(t) = s$ conditional on object $s$ being the $k$-th most-preferred object according to agents' preferences.
\begin{align*}
    \mathbb{P}[s(t) = s] 
    &= \sum_{k = 1}^m \mathbb{P}[s^k = s] \cdot \mathbb{P}[s(t) = s | s^k = s]\\
    &= \sum_{k = 1}^m \mathbb{P}[s^k = s] \cdot \mathbb{P}[s(t) \neq s^1, \ldots, s^{k-1} | s^k = s] 
    \cdot \mathbb{P}[s(t) = s^k | s^k = s] \\
    &= \frac{1}{m} \sum_{k = 1}^m \sum_{S' \in \binom{S \setminus \{s\}}{k-1}} \mathbb{P}[S^{k - 1} = S'] \cdot \mathbb{P}\left[s(t) \notin S^{k - 1} | S^{k - 1} = S' \right] \cdot \mathbb{P}[s(t) = s^k | S^{k - 1} = S'] \\
    &= \frac{1}{m} \sum_{k = 1}^m \sum_{S' \in \binom{S \setminus \{s\}}{k - 1}} \frac{1}{\binom{m - 1}{k - 1}} \cdot \mathbbm{1}\left[\sum_{s' \in S'} q_{s'} < t \right] \cdot \mathbbm{1}\left[q_{s} \ge t - \sum_{s' \in S'} q_{s'} \right] \\
    &= \frac{1}{m} \sum_{k = 1}^m \sum_{S' \in \binom{S \setminus \{s\}}{k - 1}} \frac{1}{\binom{m - 1}{k - 1}} \cdot \mathbbm{1}\left[t - q_s \leq \sum_{s' \in S'} q_{s'} < t \right] 
\end{align*} 
Where conditioning on $s^k = s$ is reflected by the domain of the second summand, $S' \in \binom{S \setminus \{s\}}{k -1}$. With unit capacities, the indicator reads $t - 1 \leq k - 1 < t \iff t \leq k < t + 1$ which is only satisfied for $k = t$. More generally, with common capacities, $ 1 \leq q = q_s, \forall s \in S$, the indicator reads $t - q \leq (k - 1)q < t \iff \frac{t}{q} \leq k < \frac{t}{q} + 1$, which only holds for a single value of $k$ (call it $k^*$), and we get
\begin{align*}
    \mathbb{P}[s(t) = s] = \frac{1}{m} \sum_{S' \in \binom{S \setminus \{s\}}{k^* - 1}} \frac{1}{\binom{m - 1}{k^* - 1}} = \frac{1}{m} \frac{\binom{m - 1}{k^* - 1}}{\binom{m - 1}{k^* - 1}} = \frac{1}{m}
\end{align*}
This matches the intuition that if objects have the same capacity and preferences are independent and random, every agent-object pair is equally likely to form a match. 

We can also verify that under unit capacities, the probability of envy should equal one. Evaluating $\mathbb{P}[s(t') \neq s(t) | s(t) = s]$ under unit capacities, the first indicator reads $k - 1 < t \leq k$ (which holds for $k = t$), while the second indicator reads $1 < t' - k + 1$ (which holds for $k < t'$). Plugging in $\mathbb{P}[s(t) = s] = \frac{1}{m}$, we get
\begin{align*}
    \mathbb{P}[s(t') \neq s(t) | s(t) = s] = \sum_{S' \in \binom{S \setminus \{s\}}{t - 1}}  \frac{\frac{1}{m} \frac{1}{\binom{m - 1}{t - 1}}}{\frac{1}{m}} = \frac{\binom{m - 1}{t - 1}}{\binom{m - 1}{t - 1}} = 1
\end{align*}
Repeating the proof of \hyperref[theorem1]{Theorem 1} with these new expressions yields
\begin{align*}
    &\operatorname*{argmin}_{\rho \in \mathcal{R}_n} \E_P \left[ N\left(\mu_\rho^{\text{SD}}(P, \succ)\right) \right] = \operatorname*{argmin}_{\rho \in \mathcal{R}_n} \sum_{s \in S} \sum_{t < t'} p(s, t', t, q) \cdot \mathbbm{1}[\rho(t') \succ_s \rho(t)]
\end{align*}
Where the weights $p(s, t', t, q)$ are equal to
\begin{align*}
        p(s, t', t, q) = \sum_{k = 1}^m \sum_{S' \in \binom{S \setminus \{s\}}{k - 1}} \frac{1}{m \binom{m - 1}{k - 1}} \cdot \mathbbm{1}\left[ \sum_{s' \in S'} q_{s'} < t \leq \sum_{s' \in S'} q_{s'} + q_s\right] \cdot \mathbbm{1}\left[ \sum_{s' \in S'} q_{s'} + q_s < t' \right]
    \end{align*}
\end{proof}

\bibliographystyle{plainnat}
\bibliography{bib}

\end{document}